\documentclass[a4paper,USenglish,cleveref, autoref, thm-restate]{lipics-v2021}

\usepackage{hyperref}

\usepackage[linesnumbered,ruled,vlined]{algorithm2e}
\SetCommentSty{mycommfont}
\usepackage{nicefrac}   
\usepackage{amsmath,cleveref}
\usepackage{autonum}
\usepackage{amsthm,amssymb,mathtools}

\usepackage{bigints}
\usepackage{comment}
\usepackage{colortbl}
\newcolumntype{d}{>{\columncolor{gray!30}}c}
\usepackage{booktabs}

\usepackage[numbers,sort]{natbib} 
\usepackage{enumitem}
\newcommand{\true}{\texttt{true}}
\newcommand{\false}{\texttt{false}}

\newcommand{\bA}{\mathbf{A}}
\newcommand{\DP}{\mathrm{DP}}
\newcommand{\vmax}{v_{\max}}
\renewcommand{\mid}{:}
\DeclareMathOperator*{\argmax}{arg\,max}
\newcommand{\ot}{\leftarrow}

\usepackage{tikz}
\usetikzlibrary{positioning,shapes.misc,calc}

\usepackage{todonotes}
\definecolor{myyellow}{cmyk}{0,0.02,0.23,0.01}

\SetKwInput{KwInput}{Input}                
\SetKwInput{KwOutput}{Output}              

\title{Simultaneously Fair Allocation of Indivisible Items Across Multiple Dimensions}

\author{Yasushi Kawase}{The University of Tokyo, Japan}{kawase@mist.i.u-tokyo.ac.jp}{https://orcid.org/0000-0001-5626-779X}{Supported by JSPS KAKENHI Grant Number JP25K00137, by JST PRESTO Grant Number JPMJPR2122, by JST ERATO Grant Number JPMJER2301, and by Value Exchange Engineering, a joint research project between Mercari, Inc.\ and the RIISE.}
\author{Bodhayan Roy}{Indian Institute of Technology Kharagpur, India}{bodhayan.roy@gmail.com}{https://orcid.org/0009-0005-6476-3060}{}
\author{Mohammad Azharuddin Sanpui}{Indian Institute of Technology Kharagpur, India}{
azharuddinsanpui123@gmail.com}{https://orcid.org/0000-0001-5030-9645}{}

\authorrunning{Y. Kawase, B. Roy, and M.A. Sanpui} 

\Copyright{Jane Open Access and Joan R. Public} 



\category{} 

\relatedversion{} 

\acknowledgements{}

\nolinenumbers 

\EventEditors{John Q. Open and Joan R. Access}
\EventNoEds{2}
\EventLongTitle{42nd Conference on Very Important Topics (CVIT 2016)}
\EventShortTitle{CVIT 2016}
\EventAcronym{CVIT}
\EventYear{2016}
\EventDate{December 24--27, 2016}
\EventLocation{Little Whinging, United Kingdom}
\EventLogo{}
\SeriesVolume{42}
\ArticleNo{23}

\begin{document}
\maketitle

\begin{abstract}
This paper explores the fair allocation of indivisible items in a multidimensional setting, motivated by the need to address fairness in complex environments where agents assess bundles according to multiple criteria. Such multidimensional settings are not merely of theoretical interest but are central to many real-world applications. For example, cloud computing resources are evaluated based on multiple criteria such as CPU cores, memory, and network bandwidth. In such cases, traditional one-dimensional fairness notions fail to capture fairness across multiple attributes.
To address these challenges, we study two relaxed variants of envy-freeness: \emph{weak simultaneously envy-free up to $c$ goods (weak sEF$c$)} and \emph{strong simultaneously envy-free up to $c$ goods (strong sEF$c$)}, which accommodate the multidimensionality of agents’ preferences.
Under the weak notion, for every pair of agents and for each dimension, any perceived envy can be eliminated by removing, if necessary, a different set of goods from the envied agent’s allocation. In contrast, the strong version requires selecting a single set of goods whose removal from the envied bundle simultaneously eliminates envy in every dimension.
We provide upper and lower bounds on the relaxation parameter $c$ that guarantee the existence of weak or strong sEF$c$ allocations, where these bounds are independent of the total number of items. In addition, we present algorithms for checking whether a weak or strong sEF$c$ allocation exists. Moreover, we establish NP-hardness results for checking the existence of weak sEF1 and strong sEF1 allocations.
\end{abstract}


\section{Introduction}


Fair allocation of resources is a central topic in economics, operations research, computer science, and public policy~\cite{Amanatidis2022FairDO,reiter1962allocating,maskin1987fair,huesch2012one}. 
Among the various notions of fairness studied in the literature, one of the most prominent is \emph{envy-freeness}. An allocation is said to be envy-free if each agent weakly prefers her own bundle to that of any other agent, according to her preferences.

While envy-freeness can always be achieved when resources are divisible~\cite{brams1995envy,aziz2016computational}, it is often unattainable for indivisible resources. For example, when allocating a single valuable indivisible item to two agents, there is no envy-free complete allocation.
This motivates an increasing amount of research to explore the relaxations of envy-freeness, such as \emph{envy-freeness to envy-freeness up to $c$ items (EF$c$)}~\cite{lipton2004approximately,suksompong2020number}.
An allocation is EF$c$ if, for any pair of agents, any envy one agent feels towards another can be eliminated by removing at most $c$ items from the latter's bundle, where $c$ is a positive integer.

Traditionally, resource allocation models have focused on trying to be fair with respect to a single attribute, such as profit or cost~\cite{aziz2022algorithmic,lipton2004approximately}.
However, in practice, the value of an item is rarely determined by a single criterion~\cite{flammini2025fair,Bu2024fair,Barman2025FairDW}. 
For example, cloud computing resources are evaluated based on multiple criteria such as CPU cores, threads, GPU capabilities, memory, disk space, network bandwidth, energy efficiency.
In e-commerce, products are evaluated by price, quality, sustainability, and user ratings.
Even in everyday scenarios---such as dividing household chores or allocating office space---multiple factors must be considered.
This multidimensionality raises a natural question: 
\begin{quote}
\emph{To what extent can relaxed notions of fairness be achieved across all attributes simultaneously, and how can such a fair allocation be computed efficiently?}
\end{quote}
In this paper, we investigate the fair allocation of indivisible items in a multidimensional setting, where each item is characterized by multiple attributes and agents have attribute-dependent preferences.

Our problem is closely related to a model recently introduced by Bu et al.~\cite{Bu2024fair}, who consider a fair resource allocation framework that accounts for the preferences of both the allocator and the agents.
Their model can be viewed as a special case of our setting, where each item has two attributes---one according to the allocator and one according to the agents.
They introduced a fairness criterion called \emph{doubly EF$c$}: an allocation is called doubly EF$c$ if it is EF$c$ for the agents and also EF$c$ from the allocator's perspective.
They provided the existence of a doubly EF1 allocation when there are two agents.
However, they showed that a triply EF1 allocation (i.e., EF$c$ from three standpoints) may not exist.

In this work, we introduce and analyze the notion of \emph{simultaneous EF$c$}, which extends the concept of doubly EF$c$ to settings with an arbitrary number of attributes, rather than being limited to two or three. In the setting of Bu et al.~\cite{Bu2024fair}, the set of items removed to eliminate envy can be chosen separately for each attribute, reflecting the perspective of different stakeholders for each attribute. However, such flexibility may not always be appropriate in our multi-attribute context. In many real-world scenarios, the attributes represent different criteria or dimensions evaluated by the same agents, rather than by distinct individuals depending on the attribute. In such cases, it is more natural and meaningful to require that the set of items removed to mitigate envy be the same across all attributes. Motivated by this observation, we define two variants of simultaneous EF$c$: one that allows attribute-specific removals and another that requires a common removal set for all attributes.

\subsection{Our Results}
We study the fair allocation of indivisible items in a multidimensional setting, where each agent's valuation is represented by an $\ell$-dimensional vector. In this framework, we explore two notions of relaxed envy-freeness: (i) \emph{weak simultaneously envy-free up to $c$ goods (weak sEF$c$)}, and (ii) \emph{strong simultaneously envy-free up to $c$ goods (strong sEF$c$)}. Notably, weak sEF$c$ in the two-dimensional case corresponds to the doubly EF$c$ and was introduced by Bu et al.~\cite{Bu2024fair}. In contrast, strong sEF$c$ is a new fairness concept introduced in this paper.

Under weak sEF$c$, for any pair of distinct agents $i$ and $i'$, and for every evaluation dimension, any envy that agent $i$ may have toward agent $i'$ can be eliminated by removing at most $c$ items from agent $i'$'s allocation in that dimension. In contrast, strong sEF$c$ requires that, for every pair of agents $i$ and $i'$, there exists a set of at most $c$ items whose removal from agent $i'$'s bundle simultaneously eliminates the envy of agent $i$ in every dimension.

Bu et al.~\cite{Bu2024fair} provided an instance with two identical agents and three dimensions that admits no weak sEF1 allocation (\Cref{thm:weak-nonex-1}).
We observe that this result can be extended to the general $\ell$-dimensional case by showing that $c=\Omega(\sqrt{\ell})$ is a lower bound to guarantee the existence of a weak sEF$c$ allocation (\Cref{thm:weak-nonex-c}).
For the strong version, we establish a stronger lower bound: namely, that $c=\lceil \ell/2\rceil$ is necessary to guarantee the existence of a strong sEF$c$ allocation (\Cref{thm:strong-nonex-c}).

Next, we show that when there are two agents, a strong sEF$(2\ell-1)$ allocation always exists and can be computed in polynomial time (\Cref{thm:nonid-ex-c}).
Furthermore, for a general number of agents, we prove the existence of a strong sEF$(n^2\ell^2)$ allocation (\Cref{thm:n-EFc}). We remark that since every strong sEF$c$ allocation is also weak sEF$c$ by definition, these existence results immediately apply to the weak case as well.
In addition, we develop algorithms to decide whether a weak or strong sEF$c$ allocation exists in time $(m\vmax)^{O(cn^2\ell)}$, where $\vmax$ denotes the maximum value assigned by any agent to any item in any dimension (\Cref{thm:dp-weak} and \Cref{thm:dp-strong}).

Finally, we complement our existence results with computational intractability findings. Specifically, we prove that even checking for the existence of a weak sEF1 allocation is NP-complete, even when there are only two identical agents with binary valuations (\Cref{thm:hard-weak}). Moreover, we demonstrate that the corresponding decision problem for strong sEF1 allocations is NP-complete even in the case of two identical agents and two dimensions (\Cref{thm:hard-strong}).
Further, we demonstrate that determining whether a given allocation is strong sEF$c$ is NP-complete even for the identical binary case, when $c$ is a parameter.

\subsection{Related Work}

Barman et al.~\cite{Barman2025FairDW} studied fair division with market values, a special case of the fair allocation with the allocator's preference.
Specifically, their model introduces a market valuation that serves as a distinct valuation profiles shared identically by all agents. This results in a two-dimensional valuation framework that comprises agents' personalized valuations alongside a universal market valuation.

A number of studies have explored fair division among groups of agents~\cite{foley1966resource,manurangsi,SegalHalevi,Suksompong2017,Conitzer2019GroupFF,berliant1992fair,Aziz2019AlmostGE,segal2019}. 
In these models, $\nu = \nu_1 + \cdots+ \nu_n$ agents will be partitioned into $n~(\geq 2)$ groups, with group $i$ comprising $\nu_i~(\geq 1)$ agents. It is crucial to point out that when each group consists an identical number of agents (i.e., $\nu_1 =\cdots = \nu_n=\nu/n$), the fair division problem under this group framework aligns with our model, in which each group can be regarded as a single agent with $\nu/n$ perspectives.
Moreover, even when the number of agents in each group are heterogeneous, we can reduce an instance of group framework to our setting with $\max_{i'\in[n]}$ dimensions by adding $\max_{i'\in[n]}\nu_{i'}-\nu_i$ dummy agents to each group $i\in[n]$ such that any item evaluates to $0$.

Kyropoulou et al.~\cite{kyroalmost} extended the classical EF$c$ concept to a group setting by requiring that an agent's envy toward another group can be eliminated by removing at most $c$ items from that group's allocation. 
When the number of agents in each group is the same, the definitions of group EF$c$ coincides with our notion of weak EF$c$.
For binary valuations, Kyropoulou et al.~\cite{kyroalmost} determined the cardinalities of groups for which a group EF1 allocation always exists. 
Moreover, they demonstrated that determining the existence of a group EF1 allocation for two groups of agents with binary valuations is NP-complete. 
Note that their reduction does not extend to the setting with \emph{identical} valuations in our setting.
In contrast, we show via a different reduction that checking the existence of a weak sEF1 allocation is NP-complete even for identical binary valuations with two agents. 

When there are two groups of agents, Segal-Halevi and Suksompong~\cite{segal2019} proved that a group EF$(\nu-1)$ allocation always exists, where $\nu=\nu_1+\nu_2$. This result implies existence of a weak EF$(2\ell-1)$ allocation in our setting with two agents and $\ell$ dimensions, although they left efficient computability. In contrast, we provide existence and efficient computability for strong EF$(2\ell-1)$.

Using the discrepancy theory, Manurangsi and Suksompong~\cite{Manurangsi2021AlmostEF} proved that a group EF$c$ allocation with $c=O(\sqrt{\nu})$ always exists. 
This implies the existence of a weak EF$c$ allocation with $c=O(\sqrt{n\ell})$ in our setting.
Moreover, they showed that $c=\Omega(\sqrt{\nu})$ is necessary to guarantee the existence of group EF$c$ allocations when there are constant number of groups. 
This implies that $c=\Omega(\sqrt{\ell})$ is necessary to guarantee the existence of weak EF$c$ allocations.
Furthermore, there are some recent works~\cite{caragiannis2025new,manurangsi2025tight} that asymptotically improve the lower bound for group fairness for group fair division, when the number of groups is super-constant.

Furthermore, our model involving multi-dimensional items share some similarities with the multi-layered cake cutting problem~\cite{Hosseini2020FairDO,igarashi2021envy,Cloutier2009TwoplayerEM,Nyman2017FairDW,Lebert2013EnvyfreeTM,kawase2025resource}.
However, unlike our setting, in the multi-layered cake cutting problem each layer of the cake at the same position must be allocated to different agents.


\section{Preliminaries}

In this paper, we address the problem of assigning $m$ items, each evaluated according to $\ell$ different criteria, to $n$ agents. We refer to this as the Multi-Dimensional Fair Allocation (MDFA) problem.

For a positive integer $n$, we denote the set $\{1,2,\dots,n\}$ by $[n]$. Let $N=[n]$ denote the set of agents and $M=\{g_1,g_2,\dots,g_m\}$ denote the set of indivisible items. 
In addition, let $L=[\ell]$ denote the set of dimensions.
Each agent $i\in N$ has a valuation function $v_i\colon M\to\mathbb{Z}^\ell_+$, which assigns to every item an $\ell$-dimensional vector of nonnegative integers.
We write $v_{ijk}=v_i(g_j)_k$ to indicate the value in the $k$\textit{th} dimension of item $g_j\in M$ for agent $i\in N$.
An instance of the MDFA problem is represented as $(N,M,L,(v_i)_{i\in N})$.

An allocation $\bA=(A_1,A_2,\dots,A_n)$ is a partition of the items $M$ (i.e., $\bigcup_{i\in N}A_i=M$ and $A_i\cap A_{i'}=\emptyset$ for any distinct agents $i,i'\in N$), where each subset $A_i\subseteq M$ is allocated to agent $i\in N$. The total valuation of agent $i$ for her allocated set $A_i$ is defined as 
$v_i(A_i)=\sum_{g_j\in A_i} v_i(g_j)$. 
Specifically, the $k$\textit{th} dimensional value of agent $i$ for the allocated set $A_i$ is defined as $v_i(A_i)_k=\sum_{g_j\in A_i} v_{ijk}$.

We call an instance of the MDFA problem \emph{binary} if $v_{ijk}\in\{0,1\}$ for every agent $i\in N$, every item $g_j\in M$, and every dimension $k\in L$. An instance is said to be \emph{identical} if every agent has the same valuation function, i.e, $v_1=v_2=\dots=v_n$.

We now introduce two fairness notions. Let $c$ be a nonnegative integer.
\begin{definition}[weak sEF$c$~\cite{Bu2024fair}] 
An allocation $\bA$ is said to be \emph{weak simultaneously envy-free up to $c$ goods (weak sEF$c$)} if for any distinct agents $i,i'\in N$ and any dimension $k\in L$, there exists a set $X\subseteq A_{i'}$ such that $|X|\le c$ and $v_i(A_i)_k\ge v_i(A_{i'}\setminus X)_k$.
\end{definition}

\begin{definition}[strong sEF$c$] 
An allocation $\bA$ is said to be \emph{strong simultaneously envy-free up to $c$ goods (strong sEF$c$)} if for any distinct agents $i,i'\in N$, there exists a set $X\subseteq A_{i'}$ such that $|X|\le c$ and $v_i(A_i)_k\ge v_i(A_{i'}\setminus X)_k$ for any dimension $k\in L$.
\end{definition}
Clearly, strong sEF$c$ implies weak sEF$c$, and weak sEF$c$ implies strong sEF$\ell c$.

Bu et al.~\cite{Bu2024fair} presented an instance in which no weak sEF1 allocation exists.
The instance consists of two agents and three items, with valuations as shown in \Cref{tbl:weak-nonex-1}.
In order to achieve EF$1$ for the first dimension, items $g_1$ and $g_2$ must be allocated to different agents.
Similarly, considering the second dimension, items $g_1$ and $g_3$ must be allocated to different agents, and for the third dimension, items $g_2$ and $g_3$ must be allocated to different agents.
However, no allocation can satisfies all these conditions simultaneously, implying that no weak sEF1 allocation exists.
\begin{theorem}[Bu et al.~\cite{Bu2024fair}]\label{thm:weak-nonex-1}
Even when there are two identical agents and the number of dimensions is three, a weak sEF1 allocation may not exist.
\end{theorem}

It is known that a weak sEF1 allocation is guaranteed to exist when there are two agents and two dimensions~\cite{Bu2024fair}.
In contrast, we now demonstrate that a strong sEF1 allocation may fail to exist even when there are two identical agents and two dimensions.
To see this, consider an instance with three items, whose valuations are given in \Cref{tab:strong-nonex-1}.
Suppose, without loss of generality, that $g_1$ is allocated to agent $1$.
Then, in order to satisfy the EF1 condition for the first dimension, the item $g_2$ must be allocated to agent $2$.
Similarly, to achieve EF1 for the second dimension, the item $g_3$ must be also allocated to agent $2$.
Thus, the resulting allocation is $\bA=(\{g_1\},\, \{g_2,g_3\})$.
However, this fails the strong sEF1 condition: Agent 1 envies agent 2 unless both of the two items are removed from agent 2's bundle. Hence, in this case, no strong sEF1 allocation exists.\footnote{The allocation satisfies weak sEF1 because envy can be eliminated by removing different items in each dimension.}
Symmetric reasoning applies if $g_1$ is initially assigned to agent 2.
\begin{theorem}\label{thm:strong-nonex-1}
Even when there are two agents and the number of dimensions is two, a strong sEF1 allocation may not exist.
\end{theorem}

\begin{table}[t]
\begin{minipage}[b]{.45\textwidth}
\centering
\begin{tabular}{c|ccc}
\toprule
item $g_j$& $g_1$ & $g_2$ & $g_3$\\
\midrule
$v_{ij1}$ & $1$   & $1$   & $0$  \\
$v_{ij2}$ & $1$   & $0$   & $1$  \\
$v_{ij3}$ & $0$   & $1$   & $1$  \\
\bottomrule
\end{tabular}
\caption{Valuations for an instance that admits no weak sEF1 allocations.}\label{tbl:weak-nonex-1}
\end{minipage}\hfill
\begin{minipage}[b]{.45\textwidth}
\centering
\begin{tabular}{c|ccc}
\toprule
item $g_j$& $g_1$ & $g_2$ & $g_3$\\
\midrule
$v_{ij1}$ & $1$   & $2$   & $0$  \\
$v_{ij2}$ & $1$   & $0$   & $2$  \\
\bottomrule
\end{tabular}
\caption{Valuations for an instance that admits no strong sEF1 allocation}\label{tab:strong-nonex-1}
\end{minipage}
\end{table}

\section{Non-Existence}
In this section, we present two non-existence results that illustrate the inherent challenges in achieving fairness under the multi-dimensional framework. In particular, for every nonnegative integer $c$, we show that even when agents have binary and identical valuations, there exist instances with only two agents where no allocation satisfies either weak or strong sEF$c$. 

We first establish that weak sEF$c$ allocations may fail to exist even in binary identical instances with only two agents, when $c=\Omega(\sqrt{\ell})$.
The construction of this instance is essentially equivalent to the one given by Manurangsi and Suksompong~\cite{Manurangsi2021AlmostEF} for fair division among two groups of agents. The idea is based on discrepancy theory~\cite{alon2016probabilistic}.
\begin{theorem}\label{thm:weak-nonex-c}
For any nonnegative integer $c$, there exists a binary identical instance with two agents and $O(c^2)$ dimensions that does not admit a weak sEF$c$ allocation.
\end{theorem}
\begin{proof}
Let $r$ be the smallest power of two that is at least $4c^2+1$.
Let $H\in\{-1,+1\}^{r\times r}$ be an Hadamard matrix of order $r$, i.e., a $\{-1,+1\}$-matrix whose row vectors are mutually orthogonal.
Note that such a matrix can be constructed.
Let $J$ denote the all-$1$ matrix of order $r$ and define $H^*=(H+J)/2$.
It is known that for any $u\in\{-1,+1\}^r$, we have $\|H^*u\|_{\infty}\ge \sqrt{r}/2$~(see, e.g., \cite[Section 13.4]{alon2016probabilistic}).

Consider the instance defined by $N=\{1,2\}$, $M=\{g_1,g_2,\dots,g_{r}\}$, and $L=[r]$.
For every agent $i\in N$, item $g_j\in M$, and dimension $k\in L$, define $v_{ijk}=H^*_{jk}$.
This instance is both binary and identical.
Note that the number of dimensions $r$ is at most $2(4c^2+1)=O(c^2)$.

For an allocation $\bA=(A_1,A_2)$, consider the column vector $u\in\{-1,+1\}^r$ defined by $u_j=+1$ if $g_j\in A_1$ and $u_j=-1$ if $g_j\in A_2$.
Then, we have
$\max_{k\in L}|v_1(A_1)_k-v_1(A_2)_k| = \|H^* u\|_{\infty} \ge \sqrt{r}/2>c$.
This means that we need to remove at least $c+1$ items to eliminate envy. This completes the proof.
\end{proof}

This theorem has a direct consequence for the strong variant. Because any strong sEF$c$ allocation must also satisfy weak sEF$c$, the impossibility of achieving weak sEF$c$ under certain conditions immediately implies that strong sEF$c$ allocations cannot exist under the same conditions. 
We further strengthen this conclusion by demonstrating that even when the number of dimensions $\ell$ is exactly $2c+1$, a strong sEF$c$ allocation may fail to exist. Thus, requiring $\ell\le 2c$, or equivalently $c\ge \lceil \ell/2\rceil$, is necessary to guarantee the existence of a strong sEF$c$ allocation.
\begin{theorem}\label{thm:strong-nonex-c}
For any nonnegative integer $c$, there exists a binary instance with identical valuations for two agents and with $2c + 1$ dimensions that does not admit a strong sEF$c$ allocation.
\end{theorem}
\begin{proof}
Consider the instance with $N=\{1,2\}$, $M=\{g_1,g_2,\dots,g_{2c+1}\}$, and $L=[2c+1]$.
For every agent $i\in N$, item $g_j\in M$, and dimension $k\in L$, define
$v_{ijk}$ to be $1$ if $j=k$ and $0$ otherwise.
This instance is both binary and identical.

Now, suppose toward a contradiction that there exists a strong sEF$c$ allocation $\bA=(A_1,A_2)$. Without loss of generality, assume that $|A_1|\le |A_2|$.
Since $|A_1|+|A_2|=|M|=2c+1$, it follows that $|A_1|\le c$ and $|A_2|\ge c+1$. 

Consider agent $1$'s evaluation of agent $2$'s bundle. 
By the definition of $v_{ijk}$, only item $g_{j^*}\in A_2$ contributes a value of 1 in dimension $j^*$. 
To eliminate envy in dimension $j^*$, the item $g_{j^*}$ must be removed from $A_2$. Therefore, in order to eliminate envy in all dimensions, every item in $A_2$ must be removed. However, since $|A_2|\ge c+1$, at least $c+1$ items must be removed, exceeding the allowable limit of $c$.

This contradiction shows that no strong sEF$c$ allocation exists for this instance.
\end{proof}

\section{Existence and Computability}
In this section, we present several positive results on the existence and efficient computability of weak and strong sEF$c$ allocations. In particular, we show existence and poynomial-time computability of a strong sEF$c$ allocation for a certain $c$ that is independent of the number of items. 
Moreover, we develop dynamic programming based algorithms that can determine the existence of a weak or strong sEF$c$ allocation in polynomial time when $c$, the number of agents $n$, and the number of dimensions $\ell$ are constants and the valuations are encoded in unary.

We first demonstrate that there exists a strong sEF$(2\ell-1)$ allocation and can be computed in polynomial time when there are two agents.
Note that, since every strong sEF allocation is also a weak sEF allocation, this result immediately implies the existence (and efficient computability) of weak sEF$(2\ell-1)$ allocations as well.
\begin{theorem}\label{thm:nonid-ex-c}
When there are two agents and $c$ is at least twice the number of dimensions minus one (i.e., $c\ge 2\ell-1$), a strong sEF$c$ allocation always exists.
Moreover, such an allocation can be computed in polynomial time.
\end{theorem}
\begin{proof}
Let $(N,M,L,(v_i)_{i\in N})$ be a given instance of the MDFA problem. 
We assume that $N=\{1,2\}$, $M=\{g_1,\dots,g_m\}$, and $L=[\ell]$.
Consider the following linear program:
\begin{align}
\max_{x\in\mathbb{R}^m}\quad & \textstyle\sum_{j=1}^m v_{1j1}(2x_j-1) \\
\text{s.t.}\quad& \textstyle\sum_{j=1}^m v_{1jk}(2x_j-1)\ge 0 &&(\forall k\in [\ell]\setminus\{1\}),\label{eq:ef1}\\
                & \textstyle\sum_{j=1}^m v_{2jk}(1-2x_j)\ge 0 &&(\forall k\in [\ell]),              \label{eq:ef2}\\
                & 0\le x_j\le 1                     &&(\forall j\in [m]).                 \label{eq:feasible}
\end{align}
A solution $x$ with a nonnegative objective yields a simultaneously envy-free fractional allocation, where $x_j$ is interpreted as the fraction of item $g_j$ allocated to agent 1.
Note that, since $x=(1/2,1/2,\dots,1/2)$ is a feasible solution, the optimal objective value is nonnegative.

Since the LP has $m$ variables, any basic solution is determined by $m$ linearly independent tight constraints.
In particular, at least $m-(2\ell-1)$ of these tight constraints must be of the form $x_j=0$ or $x_j=1$ (i.e., constraints in \eqref{eq:feasible}).
This implies that every basic solution has at least $m-2\ell+1$ coordinates that are integral (i.e., equal to $0$ or $1$). Equivalently, there can be at most $2\ell-1$ fractional coordinates.

We now construct an allocation from a basic optimal solution $x^*$.
Note that a basic optimal solution must exist since the feasible region is nonempty and bounded.
Define
\begin{align}
&I_1\coloneqq \{g_j\in M\mid x^*_j=1\},
&&I_2\coloneqq \{g_j\in M\mid x^*_j=0\},\\
&F_1\coloneqq \{g_j\in M\mid 1/2\le x^*_j<1\},
&&F_2\coloneqq \{g_j\in M\mid 0< x^*_j<1/2\}.
\end{align}
Then, since all fractional coordinates of $x^*$ lie in $F_1\cup F_2$, we have $|F_1|+|F_2|\le 2\ell-1$.

Define the allocation $\bA=(A_1,A_2)$ by $A_1=I_1\cup F_1$ and $A_2=I_2\cup F_2$.
Since $x^*$ satisfies \eqref{eq:ef1} and the objective value is nonnegative, for every dimension $k\in[\ell]$ we have
\begin{align}
\textstyle
v_1(A_1)_k
= \sum_{g_j\in A_1}v_{1jk}
&\textstyle\ge \sum_{g_j\in A_1}v_{1jk}(2x^*_j-1)\\
&\textstyle\ge \sum_{g_j\in A_2}v_{1jk}(1-2x^*_j)
\ge \sum_{g_j\in I_2} v_{1jk}
= v_1(A_2\setminus F_2)_k,
\end{align}
where the second inequality holds since, for the optimal solution $x^*$, the objective value of LP is nonnegative and inequalities in \eqref{eq:ef1} are satisfied. 
Similarly, by \eqref{eq:ef2},
for every dimension $k\in[\ell]$ we have
\begin{align}
\textstyle
v_2(A_2)_k
= \sum_{g_j\in A_2}v_{2jk}
&\textstyle\ge \sum_{g_j\in A_2}v_{2jk}(1-2x^*_j)\\
&\textstyle\ge \sum_{g_j\in A_1}v_{2jk}(2x^*_j-1)
\ge \sum_{g_j\in I_1} v_{2jk}
= v_2(A_1\setminus F_1)_k.
\end{align}
Thus, any envy can be eliminated by removing all the fractional items from the bundle of the other agent. Hence, the allocation is strong sEF$(2\ell-1)$.
Since $c\ge 2\ell-1$, the allocation is strong sEF$c$ as well.

Finally, since the LP has polynomial number of constraints and variables, a basic optimal solution can be computed in polynomial time via standard linear programming techniques (see, e.g., a textbook \cite{bertsimas1997introduction}). 
Hence, the resulting allocation $\bA$ can be found in polynomial time.
\end{proof}

If the valuations of the two agents are identical, we can reduce the requirement of $2\ell-1$ to $\ell$.
Due to space constraints, the proof is deferred to the Appendix.
\begin{restatable}{theorem}{THMidExC}\label{thm:id-ex-c}
When there are two identical agents and $c$ is at least the number of dimensions (i.e., $c\ge \ell$), a strong sEF$c$ allocation always exists.
Moreover, such an allocation can be computed in polynomial time.
\end{restatable}

When the number of agents is three or more, a direct application of the polytope-based approach presented in \Cref{thm:nonid-ex-c} and \Cref{thm:id-ex-c} fails. Indeed, consider the case with three agents in one dimension, where the valuations are given by
\begin{align}
(v_{i11},v_{i21},v_{i31},\dots,v_{im1})=(2(m-1),1,1,\dots,1)
\qquad(i=1,2,3).
\end{align}
In this case, allocating half of item $g_1$ to agents $1$ and $2$, while assigning all the remaining items to agent $3$ yields an envy-free fractional allocation.
A reasonable rounding method in this situation is to assign item $g_1$ to either agent $1$ or $2$.
However, the agent who does not receive $g_1$ will envy agent $3$, and this envy cannot be eliminated by removing $m-2$ items.

To overcome this difficulty, we pre-assign certain items before employing the polytope-based approach. This strategy allows us to derive an upper bound on $c$ that guarantees the existence of a strong sEF$c$ allocation independent of the total number of items $m$.
\begin{theorem}\label{thm:n-EFc}
When $c$ is at least the square of the product of the numbers of agents and the number of dimensions (i.e., $c\ge n^2\ell^2$), a strong sEF$c$ allocation always exists.
Moreover, such an allocation can be computed in polynomial time.
\end{theorem}
\begin{proof}
Let $(N,M,L,(v_i)_{i\in N})$ be a given instance of the MDFA problem. 
We assume that $N=[n]$, $M=\{g_1,\dots,g_m\}$, and $L=[\ell]$.
If $m\le n\cdot c$, then simply allocating at most $c$ items to each agent yields a strong EF$c$ allocation.
Hence, in what follows, we assume that $m>n\cdot c~(\ge n^3\ell^2)$.

For each agent $i\in N$ and dimension $k\in L$, we begin by sequentially allocating the $(n-1)^2\ell$ most desired items in dimension $k$ to agent $i$.
Denote by $Z_{i,k}$ the set of items allocated to agent $i$ with respect to dimension $k$. This procedure is described in Algorithm~\ref{alg:n-EFc}.
Intuitively, the set $Z_{i,k}$ is used to compensate for the loss of value in the $k$th dimension that occurs when rounding a fractional solution.
Let $R\coloneqq M\setminus\bigcup_{i\in N}\bigcup_{k\in L}Z_{i,k}$ be the set of remaining items.
Note that, for every $g_j\in Z_{i,k}$ and $g_{j'}\in R$, we have $v_{ijk}\ge v_{ij'k}$.

\begin{algorithm}
\caption{Pre-assignment of items}\label{alg:n-EFc}
Let $R\ot M$\;
\ForEach{$i\in N$ and $k\in L$}{
    Initialize $Z_{i,k}\ot \emptyset$\;
    \For{$p\gets 1$ \KwTo $(n-1)^2\ell$}{
        Choose $g\in\argmax_{g_j\in R}v_{ijk}$\;
        Update $Z_{i,k}\ot Z_{i,k}\cup\{g\}$ and $R\ot R\setminus\{g\}$\;
    }    
}
\end{algorithm}

Consider the following polytope $P$ in $\mathbb{R}^{N\times R}$:
\begin{align}
P\coloneqq\left\{
x\in\mathbb{R}^{N\times R}
\mid
\begin{array}{ll}
\sum_{g\in R} v_{i}(g)_k (x_{ij}-x_{i'j})\ge 0 & (\forall i,i'\in N \text{ with }i\ne i',\ \forall k\in L),\\[2pt]
\sum_{i\in N} x_{ig}=1                         & (\forall g\in R),\\[2pt]
x_{ig}\ge 0                                    & (\forall i\in N,\ \forall g\in R)
\end{array}\!\!%
\right\}.
\end{align}
A feasible point $x\in P$ is a simultaneously proportional fractional allocation of $R$, where $x_{ig}$ is interpreted as the fraction of item $g$ allocated to agent $i$.
Note that, since the uniform allocation $x=(1/n)_{i\in N, g\in R}$ is feasible, the polytope $P$ is nonempty.

Let $m'=|R|$.
Since $P$ is in $m'n$-dimensional space, any extreme point is determined by $m'n$ linearly independent tight constraints.
In particular, at least $m'n-m'-n(n-1)\ell$ of these tight constraints must be zero.
This implies that every extreme point has at most $m'+n(n-1)\ell$ nonzero coordinates.
Consequently, at least $m'-n(n-1)\ell$ items are allocated integrally (i.e., for such an item $g\in R$, there exists some agent $i$ with $x_{ig}=1$).

Now, let $x^*$ be an extreme point of $P$, and for each agent $i\in N$, define $I_i\coloneqq\{g\in R\mid x^*_{ig}=1\}$.
Let $F\coloneqq R\setminus \bigcup_{i\in N}I_i$ be the set of items that are not allocated integrally.
By the above discussion, we have $|F|\le n(n-1)\ell$.
Partition $F$ arbitrary into $n$ subsets $(F_1,\dots,F_n)$ such that $|F_i|\le (n-1)\ell$ for all $i\in N$.
Then, define the allocation $\bA=(A_1,\dots,A_n)$ by $A_i=\left(\bigcup_{k\in L}Z_{i,k}\right)\cup I_i\cup F_i$ for each $i\in N$.
We now show that $\bA$ is a strong sEF$c$ allocation if $c\ge n^2\ell^2$. 

For every $i,i'\in N$ and every $k\in L$, we have
\begin{align}
v_i(A_i)_k
&\ge v_i(I_i)_k+v_i(F_i)_k+v_i(Z_{i,k})_k
\ge v_i(I_{i'})_k-v_i(F)_k+v_i(F_i)_k+v_i(Z_{i,k})_k\\
&= \textstyle v_i\left(A_{i'}\setminus \left(F_{i'}\cup\bigcup_{k'\in L}Z_{i',k'}\right)\right)_k-v_i(F\setminus F_i)_k+v_i(Z_{i,k})_k\\
&\ge \textstyle v_i\left(A_{i'}\setminus \left(F_{i'}\cup\bigcup_{k'\in L}Z_{i',k'}\right)\right)_k,
\end{align}
where the last inequality follows from the facts that, for every $g_j\in Z_{i,k}$ and $g_{j'}\in F\subseteq R$, $v_{ijk}\ge v_{ij'k}$, and $|F\setminus F_i|\le (n-1)^2\ell\le |Z_{i,k}|$.
Therefore, we obtain that $\bA$ is a strong sEF$c$ when $c=n^2\ell^2$ since 
\begin{align}
\textstyle
n^2\ell^2
\ge (n-1)\ell+\ell\cdot (n-1)^2\ell
\ge |F_{i'}|+\sum_{k'\in L}|Z_{i',k'}|
=\textstyle\big|F_{i'}\cup\bigcup_{k'\in L}Z_{i',k'}\big|
\end{align}
for every $i'\in N$.

Finally, since the polytope $P$ has polynomial number of constraints and variables, a basic optimal solution can be computed in polynomial time via standard linear programming techniques.
Hence, the resulting allocation $\bA$ can be found in polynomial time.
\end{proof}

Next, we provide a dynamic programming algorithm that checks existence of a weak sEF$c$ allocation.
The key idea is to process the items one by one, gradually building up partial allocations and summarizing their ``state'' in a way that keeps track of all the information necessary to verify the weak sEF$c$ condition later.
\begin{theorem}\label{thm:dp-weak}
The existence of a weak sEF$c$ allocation can be checked in time $m^{O(n^2\ell)}\cdot (\vmax+1)^{O(cn^2\ell)}$, where 
$\vmax\coloneqq \max_{i\in N}\max_{g_j\in M}\max_{k\in L}v_{ijk}$.
In particular, if $c$, $n$, and $\ell$ are constant and the valuations are encoded in unary, the existence of a weak sEF$c$ can be checked in polynomial time.
Similarly, if $\ell$ and $n$ are constant, $c=O(\log m)$, and the valuations are binary, then the existence of a strong sEF$c$ can be checked in polynomial time.
\end{theorem}
\begin{proof}
Let $(N,M,L,(v_i)_{i\in N})$ be an instance of the MDFA problem. 
We assume that $N=[n]$, $M=\{g_1,\dots,g_m\}$, $L=[\ell]$, and that for each agent $i\in N$ and item $g\in M$.
For each $j\in\{0,1,\dots,m\}$, let $M_j=\{g_1,g_2,\dots,g_j\}$ denote the prefix of items. 
Without loss of generality, we assume that $m>n\cdot c$ since otherwise simply allocating at most $c$ items to each agent yields a weak EF$c$ allocation.

For each $j\in \{0,1,\dots,m\}$ and for each allocation $\bA$ of $M_j$, we summarize its state $(V,T)$ as follows:
\begin{itemize}
\item $V\in \mathbb{Z}^{N\times N\times L}$ represents the valuation profile. For every $i,i'\in N$ and every $k\in L$, $V_{i,i',k}=v_{i}(A_{i'})_k\in [0,1,\dots,m\cdot \vmax]$. 

\item $T\in \mathbb{Z}^{N\times N\times L\times [c]}$ encodes the top $c$ item values. For every $i,i'\in N$, every $k\in L$, and $s\in [c]$, let $T_{i,i',k,s}$ denote the $s$th largest value in the multiset $\{v_i(g)\mid g\in A_{i'}\}$.
If $|A_{i'}|<c$, the remaining entries in $T_{i,i',k}$ are filled with $0$.
\end{itemize}

We now describe a dynamic programming algorithm that computes a table $\DP[j][V][T]$, which is a Boolean indicator that is set to $\true$ if there exists an allocation of the items in $M_j$ that achieves the state $(V,T)$, and $\false$ otherwise. Notice that the size of the DP table is at most
\begin{align}
(m+1)\cdot (m\cdot \vmax+1)^{n\cdot n\cdot \ell}\cdot (\vmax+1)^{c\cdot n\cdot n\cdot \ell}=m^{O(n^2\ell)}\cdot (\vmax+1)^{O(cn^2\ell)}.
\end{align}

\medskip
\noindent\textbf{Base case.}
For $j=0$ (i.e., no items are allocated), the table $\DP[0][V][T]$ is $\true$ only when the valuations $V$ and the top $c$ item values $T$ are all zero.
In other words, the DP table can be filled as follows:
\begin{align}
\DP[0][V][T]&=
\begin{cases}
\true  &\text{if }V_{i,i'k}=T_{i,i',k,s}=0\text{ for all $i,i'\in N$, $k\in L$, and $s\in[c]$},\\
\false &\text{otherwise}.
\end{cases}
\end{align}

\medskip
\noindent\textbf{Inductive Step.}
For $j\ge 1$, assume that $\DP[j-1][V][T]=\true$ for some state $(V,T)$ corresponding to an allocation of $M_{j-1}$. We then process item $g_j$. For each agent $i\in N$, consider allocating $g_j$ to $i$. For each such choice, we update the state as follows:
\begin{itemize}
\item \textbf{Update $V$:} For every $i',i''\in N$ and every $k\in L$, define
\begin{align}
V'_{i',i'',k}&=\begin{cases}
V_{i',i'',k} + v_{i'jk}  & \text{if }i''=i,\\
V_{i',i'',k}             & \text{otherwise}.
\end{cases}
\end{align}
This corresponds to the valuation profile of the allocation where item $g_j$ is additionally assigned to agent $i$.

\item \textbf{Update $T$:} For every $i'\in N$ and every $k\in L$, define $T'_{i',i,k}$ by updating $T_{i',i,k}$ by inserting the value $v_{i'jk}$ into the multiset $\{T_{i',i,k,1},T_{i',i,k,2},\dots,T_{i',i,k,c}\}$ and retaining only the $c$ highest entries. For every agent $i''\neq i$, $T'_{i',i'',k}$ is defined to be $T_{i',i'',k}$.
\end{itemize}
Note that this update can be performed in $O(n^2\cdot\ell\cdot c)$ time.

After performing these updates for every state $(V,T)$ with $\DP[j-1][V][T]=\texttt{true}$ and for every choice of agent $i\in N$ who receives $g_j$, we set the corresponding entry $\DP[j][V'][T']$ to $\true$. All other cells are set to $\texttt{false}$.

\medskip
After processing all items, the table $\DP[m][\,\cdot\,][\,\cdot\,]$ represents possible states corresponding to full allocations of $M$. We then examine each state $(V,T)$ such that $\DP[m][V][T]=\true$ to verify the weak sEF$c$ condition. Specifically, for every pair of distinct agents $i,i'\in N$ and for every dimension $k\in L$, weak sEF$c$ requires that there exists a set $X\subseteq A_{i'}$, with $|X|\le c$, such that
$v_i(A_i)_k\ge v_i(A_{i'})_k-\sum_{g\in X}v_i(g)_k$.
Since the list $T_{i,i',k}$ contains the $c$ highest values among the multiset $\{v_i(g)_k : g\in A_{i'}\}$, it suffices to verify that
$V_{iik}\ge V_{ii'k}-\sum_{s=1}^c T_{i,i',k,s}$
for every distinct $i,i'\in N$ and every $k\in L$. This condition can be checked in $O(n^2\cdot \ell\cdot c)$ time.
A weak sEF$c$ allocation exists if and only if at least one state $(V,T)$ such that $\DP[m][V][T]=\true$ satisfies this inequality.
Finally, the overall running time is 
\begin{align}
m^{O(n^2\ell)}\cdot (\vmax+1)^{O(cn^2\ell)}\cdot n\cdot O(n^2\cdot \ell\cdot c) = 
m^{O(n^2\ell)}\cdot (\vmax+1)^{O(cn^2\ell)},
\end{align}
since computing the DP table dominates the running time.
This completes the proof.
\end{proof}

Finally, we provide a dynamic programming algorithm that checks the existence of a strong sEF$c$ allocation.
Unlike the approach for weak sEF$c$, we first enumerate all possible combinations of subsets used to eliminate envies.
\begin{theorem}\label{thm:dp-strong}
The existence of a strong sEF$c$ allocation can be checked in time $m^{O(n^2(c+\ell))}\cdot (\vmax+1)^{O(n^2\ell)}$, where $\vmax\coloneqq \max_{i\in N}\max_{g_j\in M}\max_{k\in L}v_{ijk}$.
In particular, if $c$, $n$, and $\ell$ are constant and the valuations are encoded in unary, then the existence of a strong sEF$c$ can be checked in polynomial time.
\end{theorem}
\begin{proof}
Let $(N,M,L,(v_i)_{i\in N})$ be an instance of the MDFA problem. 
We assume that $N=[n]$, $M=\{g_1,\dots,g_m\}$, $L=[\ell]$, and that for each agent $i\in N$ and item $g\in M$.
Without loss of generality, we assume that $m>n\cdot c$ since otherwise simply allocating at most $c$ items to each agent yields a strong EF$c$ allocation.

We enumerate a collection $(X_{ii'})_{i,i'\in N}$ with $|X_{ii'}|\le c$ for all $i,i'\in N$.
Each $X_{ii'}$ represents the subset that is used to eliminate envy from agent $i$ towards agent $i'$.
We consider allocations $\bA$ that satisfies $A_i\supseteq \bigcup_{i'\in N}X_{ii'}$ for all agent $i\in N$.
To satisfy consistency, for distinct agents $i$ and $\iota$, the sets of items used for envy elimination must be disjoint, i.e., 
$\left(\bigcup_{i'\in N}X_{ii'}\right)\cap \left(\bigcup_{\iota'\in N}X_{\iota\iota'}\right)=\emptyset$.
Note that the number of possible choices for the collection $(X_{ii'})_{i,i'\in N}$ is at most $\prod_{i,i'\in N} (1+m)^c  = m^{O(cn^2)}$.

Fix some collection $(X_{ii'})_{i,i'\in N}$ satisfying $|X_{ii'}|\le c$ for all $i,i'\in N$ and $X_{ii'}\cap X_{\iota\iota'}=\emptyset$ for all $i,i',\iota,\iota'$ with $i\ne \iota$.
Define $X_i=\bigcup_{i'\in N}X_{ii'}$ for each $i\in N$ and $X=\bigcup_{i\in N}X_{i}$. Let $M'=M\setminus X$ be the set of remaining items.
Write $M'$ as $\{g_{\sigma(1)},\dots,g_{\sigma(m')}\}$ and for each $j\in\{0,1,\dots,m'\}$, let $M'_j=\{g_{\sigma(1)},g_{\sigma(2)},\dots,g_{\sigma(j)}\}$ denote the prefix of items.

For each $j\in \{0,1,\dots,m'\}$ and for each allocation $\bA$ of $X\cup M'_j$ with $A_i\supseteq \bigcup_{i'\in N}X_{ii'}$ for all $i\in N$, we keep the valuation profile $V=(v_{i}(A_{i'})_k)_{i,i'\in N,\,k\in L}$.

We now describe a dynamic programming algorithm that computes a table $\DP[j][V]$, which is a Boolean indicator that is set to $\true$ if there exists an allocation of the items in $X\cup M'_j$ that achieves the valuation profile $V$, and $\false$ otherwise. Notice that the size of the DP table is at most
$(m+1)\cdot (m\cdot \vmax+1)^{n\cdot n\cdot \ell}=m^{O(n^2\ell)}\cdot (\vmax+1)^{O(n^2\ell)}$.

\medskip
\noindent\textbf{Base case.}
For $j=0$ (i.e., no items are allocated), the table $\DP[0][V]$ is $\true$ only when the valuation profile $V$ correspond to the allocation $(X_1,\dots,X_n)$.
In other words, the DP table can be filled as follows:
\begin{align}
\DP[0][V]&=
\begin{cases}
\true  &\text{if }V_{i,i'k}=\sum_{g_j\in X_{i'}}v_{ijk}\text{ for all $i,i'\in N$, $k\in L$},\\
\false &\text{otherwise}.
\end{cases}
\end{align}

\medskip
\noindent\textbf{Inductive Step.}
For $j\ge 1$, assume that $\DP[j-1][V]=\true$ for some valuation profile $V$ corresponding to an allocation of $X\cup M'_{j-1}$. We then process item $g_{\sigma(j)}$. For each agent $i\in N$, consider allocating $g_{\sigma(j)}$ to $i$. For each such choice, we construct an updated valuation profile $V'$ as follows:
\begin{align}
V'_{i',i'',k}&=\begin{cases}
V_{i',i'',k} + v_{i'jk}  & \text{if }i''=i,\\
V_{i',i'',k}             & \text{otherwise},
\end{cases}
\quad\text{for every $i',i''\in N$ and every $k\in L$}.
\end{align}
This corresponds to the valuation profile of the allocation where item $g_{\sigma(j)}$ is additionally assigned to agent $i$.
After performing these updates for every valuation profile $V$ with $\DP[j-1][V]=\texttt{true}$ and for every choice of agent $i\in N$ who receives $g_{\sigma(j)}$, we set the corresponding entry $\DP[j][V']$ to $\true$. All other cells are set to $\texttt{false}$.

\medskip
After processing all items, the table $\DP[m][\,\cdot\,]$ represents possible valuation profile corresponding to full allocations of $M$. We then examine each valuation profile $V$ such that $\DP[m][V]=\true$ to verify the strong sEF$c$ condition. Specifically, for every pair of distinct agents $i,i'\in N$ and for every dimension $k\in L$, it suffices to verify that
$V_{iik}\ge V_{ii'k}-\sum_{j\in g_j\in X_{i,i'}} v_{ijk}$.

This condition can be checked in $O(n^2\cdot \ell\cdot c)$ time.
A strong sEF$c$ allocation $\bA$ corresponding to $(X_{ii'})_{i,i'\in N}$ exists if and only if at least one state $V$ such that $\DP[m][V]=\true$ satisfies this inequality.

Finally, the overall running time is 
\begin{align}
m^{O(cn^2)}\cdot m^{O(n^2\ell)}\cdot (\vmax+1)^{O(n^2\ell)}\cdot n\cdot O(n^2\cdot \ell\cdot c) = m^{O(n^2(c+\ell))}\cdot (\vmax+1)^{O(n^2\ell)}.
\end{align}
This completes the proof.
\end{proof}


\section{Hardness of Checking Existence}
In this section, we show that checking the existence of a weak and a strong sEF1 allocation are NP-hard.
Moreover, we demonstrate the NP-hardness of checking whether a given allocation is strongly sEF$c$.

We first prove that checking whether a weak sEF1 allocation exists is strongly NP-complete, even in the case of identical binary valuations with two agents.
\begin{theorem}\label{thm:hard-weak}
Checking the existence of a weak sEF1 allocation is strongly NP-complete even for identical binary case with two agents. 
\end{theorem}
\begin{proof}
The problem is in NP, since we can verify the condition of weak sEF1 for a given allocation in polynomial time.

We prove strong NP-hardness by reduction from the monotone not-all-equal 3-SAT (MNAE3SAT) problem, which is known to be strongly NP-complete~\cite{schaefer1978complexity}. 
In the problem, we are given $n$ variables $x_1,x_2,\dots,x_n$ and $m$ clauses $C_1,C_2,\dots,C_m$.
Each clause $C_k$ is of the form $C_k=(z_{k,1}, z_{k,2}, z_{k,3})$, where every literal $z_{k,1},z_{k,2},z_{k,3}$ is one of the variables $x_1,x_2,\dots,x_n$.
The goal is to decide whether there exists a truth assignment to the variables such that, in each clause, the literals do not all share the same truth value; in other words, each clause must contain at least one true literal and at least one false literal.

From a given instance of the MNAE3SAT problem, we construct an instance of the MDFA problem as follows.
Let $N=\{1,2\}$, $M=\{g_1,g_2,\dots,g_n\}$, and $L=\{1,2,\dots,m\}$.
For every agent $i\in N$, item $g_j\in M$, and dimension $k\in L$, define
\begin{align}
v_{ijk}=\begin{cases}
1 & \text{if }x_j\in\{z_{k,1},z_{k,2},z_{k,3}\},\\
0 & \text{otherwise}.
\end{cases}
\end{align}
Notice that in every dimension $k$ (which corresponds to clause $C_k$), exactly the three items corresponding to the three variables appearing in $C_k$ are assigned value $1$, and all other items are assigned $0$. Moreover, both agents have the same binary valuation functions.

We now show that the constructed instance has a weak sEF1 alocation if and only if the MNAE3SAT instance is a yes-instance.

\medskip
\noindent\textbf{($\boldsymbol{\Rightarrow}$)}
Suppose that a weak sEF1 allocation $\bA=(A_1,A_2)$ exists for the reduced instance. In the constructed instance, for each dimension $k\in L$, exactly three items (corresponding to the three variables of clause $C_k$) have value $1$. Observe that if, for some dimension $k$, all these three items are allocated to the same agent, then one agent would have a value of $0$ while the other has a value of $3$ in that dimension. Removing one item (value $1$) from the bundle with three would yield a value of $2$, which would still leave an envy. Thus, to satisfy the weak sEF1 condition, it must be that in every dimension $k$ both agents receive at least one item with value $1$.

Now define a truth assignment for the MNAE3SAT instance by assigning
$x_j$ to be $\true$ if $g_j\in A_1$ and $\false$ otherwise.
Since for every clause $C_k$ the corresponding three items (with value $1$ in dimension $k$) are split between the two agents, the clause contains at least one literal that is true and at least one that is false. Hence, the constructed truth assignment satisfies every clause in the MNAE3SAT instance.

\medskip
\noindent\textbf{($\boldsymbol{\Leftarrow}$)}
Conversely, suppose that the given MNAE3SAT instance is a yes-instance.
Let $I\subseteq [n]$ be the set of indices of variables assigned \emph{true} in some satisfying truth assignment, so that in each clause at least one variable is true and at least one is false. We construct an allocation $\bA=(A_1,A_2)$ by letting
$A_1=\{g_j\in M\mid j\in I\}$ and $A_2=\{g_j\in M\mid j\notin I\}$.
Consider any dimension $k\in L$ (i.e., any clause $C_k$). Since $C_k=(z_{k,1}, z_{k,2}, z_{k,3})$, the three items with value $1$ in dimension $k$ are precisely those items corresponding to the variables in $C_k$. Because the truth assignment is satisfying in the MNAE3SAT sense, at least one of these variables is true and at least one is false. Hence, in dimension $k$ the items with value $1$ are split between $A_1$ and $A_2$. In particular, one bundle will contain one item valued $1$ on dimension $k$ while the other will contain two.

Now, note that under our binary valuation, if an agent’s bundle in dimension $k$ yields a total value of $1$ while the other agent’s bundle yields $2$, then there is envy. However, by removing a single item (with value $1$ in dimension $k$) from the envied bundle, the value drops from $2$ to $1$, thus eliminating the envy in that dimension. Since this argument applies to every dimension (clause), the allocation $\bA$ is weak sEF1.

\medskip
Since the reduction can be performed in polynomial time and the existence of a weak sEF1 allocation is equivalent to the satisfiability of the given MNAE3SAT instance, checking the existence of a weak sEF1 allocation is NP-hard even when there are two agents.
\end{proof}

Next, we prove that verifying the existence of a strong sEF1 allocation is NP-complete, even for the identical binary case with two agents and two dimensions.
\begin{theorem}\label{thm:hard-strong}
Checking the existence of a strong sEF1 allocation is weakly NP-complete even when there are only two identical agents and the number of dimensions is two.
\end{theorem}
\begin{proof}
The problem is in NP, since we can verify the condition of strong sEF1 for a given allocation in polynomial time.

We prove NP-hardness by reduction from the PARTITION problem, which is known to be NP-complete~\cite{garey1979computers}. 
In the Partition problem, we are given $n$ positive integers $a_1,a_2,\dots,a_n$ and an integer $T$ such that $\sum_{j=1}^{n}a_j=2T$. The goal is to determine whether there exists a partition $(S_1,S_2)$ of the set $[n]$ such that $\sum_{j\in S_1}a_j=\sum_{j\in S_2}a_j=T$. 

Given an instance of the PARTITION problem, we construct an instance of the MDFA problem as follows.
Let $N=\{1,2\}$, $M=\{g_1,g_2,\dots,g_{n+2}\}$, and $L=\{1,2\}$.
Both agent have identical valuation functions.
For each item $g_j$ with $j\in[n]$, the value is $a_j$ in both dimensions.
Item $g_{n+1}$ has a very high value $2T+1$ in the first dimension and zero in the second, while item $g_{n+2}$ has a value of $2T+1$ for the second dimension and zero in the first.
The valuations are summarized in \Cref{tbl:hard-strong}. 

\begin{table}[ht]
\centering
\begin{tabular}{c|cccccc}
\toprule
             & $g_1$ & $g_2$ & $\cdots$ & $g_n$ & $g_{n+1}$ & $g_{n+2}$\\
\midrule
$v_i(g_j)_1$ & $a_1$ & $a_2$ & $\cdots$ & $a_n$ & $2T+1$    & $0$      \\
$v_i(g_j)_2$ & $a_1$ & $a_2$ & $\cdots$ & $a_n$ & $0$       & $2T+1$   \\
\bottomrule
\end{tabular}
\caption{The valuations for the reduced instance in \Cref{thm:hard-strong}.}\label{tbl:hard-strong}
\end{table}

We now show that the constructed instance has a strong sEF1 allocation if and only if the PARTITION instance is a yes-instance.

\medskip
\noindent\textbf{($\boldsymbol{\Rightarrow}$)}
Suppose that there exists a strong sEF1 allocation $\bA=(A_1,A_2)$ for the constructed instance.
Note that if both extra items $g_{n+1}$ and $g_{n+2}$ were assigned to the same agent, then to eliminate envy the other agent would need to have both items removed, which violates the strong sEF1 condition.
Thus, $g_{n+1}$ and $g_{n+2}$ must be assigned to different agents. Without loss of generality, assume that agent $1$ receives $g_{n+1}$ and agent $2$ receives $g_{n+2}$.
Define
\begin{align}
\textstyle
R_1=\sum_{g_j\in A_1\setminus \{g_{n+1}\}}a_j
\qquad\text{and}\qquad
R_2=\sum_{g_j\in A_2\setminus \{g_{n+2}\}}a_j.
\end{align}
Then, the total values for each agent in each dimension are
\begin{align}
&v_1(A_1)_1=v_2(A_1)_1=2T+1+R_1, && v_1(A_1)_2=v_2(A_1)_2=R_1,\\
&v_1(A_2)_1=v_2(A_2)_1=R_2,      && v_1(A_2)_2=v_2(A_2)_2=2T+1+R_2.
\end{align}

To eliminate envy from agent 1 to agent 2 in the second dimension, it holds that $v_1(A_1)_2\ge \min_{g\in A_2}v_1(A_2\setminus\{g\})_2$. Thus, we have
\begin{align}
\textstyle
R_1=v_1(A_1)_2\ge \min_{g\in A_2}v_1(A_2\setminus\{g\})_2=v_1(A_2\setminus\{g_{n+2}\})_2=R_2.
\end{align}
Similarly, to eliminate envy from agent 2 to agent 1 in the first dimension, we have
\begin{align}
\textstyle
R_2=v_2(A_2)_1\ge \min_{g\in A_1}v_2(A_1\setminus\{g\})_1=v_2(A_1\setminus\{g_{n+1}\})_1=R_1.
\end{align}
Together these yield $R_1=R_2=T$.
This means that the corresponding partition $(\{j\in[n]\mid g_j\in A_1\},\ \{j\in[n]\mid g_j\in A_2\})$ is a solution of the PARTITION instance.

\medskip
\noindent\textbf{($\boldsymbol{\Leftarrow}$)}
Conversely, suppose that the given PARTITION instance is a yes-instance.
Then, there exists a partition $(S_1,S_2)$ of $[n]$ such that $\sum_{j\in S_1}a_j=\sum_{j\in S_2}a_j=T$.
Define an allocation $\bA=(A_1,A_2)$ by $A_1=\{g_j\mid j\in S_1\}\cup\{g_{n+1}\}$ and $A_2=\{g_j\mid j\in S_1\}\cup\{g_{n+2}\}$.
Under this allocation, the total value for each agent in each dimension is:
\begin{align}
&v_1(A_1)_1=v_2(A_1)_1=3T+1, && v_1(A_1)_2=v_2(A_1)_2=T,\\
&v_1(A_2)_1=v_2(A_2)_1=T,    && v_1(A_2)_2=v_2(A_2)_2=3T+1.
\end{align}
Let us verify the strong sEF1 conditions:
\begin{align}
&v_1(A_1)_1\ge v_1(A_2)_1,                     && v_1(A_1)_2\ge v_1(A_2\setminus\{g_{n+2}\})_2,\\
&v_2(A_2)_1=v_2(A_1\setminus\{g_{n+1}\})_1,    && v_2(A_2)_2\ge v_2(A_1)_2.
\end{align}
Thus, removing a single item suffices to eliminate envy in every case, so the allocation is strong sEF1.

\medskip
Since the reduction can be performed in polynomial time and the existence of a strong sEF1 allocation is equivalent to the answer of the given PARTITION instance, checking the existence of a strong sEF1 allocation is NP-hard even when there are two identical agents and two dimensions.
\end{proof}

\begin{remark}
A reduction analogous to that used in the proof of \Cref{thm:hard-strong} shows that if checking the existence of an EF allocation is NP-hard in a certain one-dimensional setting, then the corresponding problem of checking the existence of a strong EF$c$ allocation in $nc$ dimensions is also NP-hard.
Specifically, given an instance $(N=[n],M=\{g_1,\dots,g_m\},(v_i)_{i\in N})$ of the one-dimensional EF allocation existence problem, we construct an instance of the MDFA problem $(\hat{N},\hat{M},\hat{L},(\hat{v}_i)_{i\in N})$ as follows:
set $\hat{N}=N$, $\hat{M}=M\cup\{g_{m+1},\dots,g_{m+nc}\}$, $\hat{L}=[nc]$, and define
\begin{align}
\hat{v}_{ijk}=\begin{cases}
v_{ij} & \text{if }j\le m,\\
\sum_{j'=1}^mv_{ij'}+1       & \text{if }j>m\text{ and }j-m=k,\\
0       & \text{otherwise},
\end{cases}
\quad\text{for all } i\in\hat{N}, g_j\in\hat{M}, k\in\hat{L}.
\end{align}
Using this reduction, we can, for example, prove that checking the existence of a strong sEF$c$ allocation is strongly NP-hard even when the valuations are identical (since it is not difficult to deduce strong NP-hardness for the one-dimensional EF allocation existence problem of identical case by a reduction from the 3-PARTITION problem).
\end{remark}

Finally, we show that determining whether a given allocation is strong sEF$c$ is NP-complete when $c$ is a parameter. Due to space constraints, the proof is deferred to the Appendix.
\begin{restatable}{theorem}{THMsch}\label{thm:strong-check-hard}
Checking whether a given allocation is strong sEF$c$ is strongly NP-complete, even for the identical binary case with two agents, when $c$ is part of the input.
\end{restatable}

\section{Conclusion}
In this paper, we introduced the problem of allocating indivisible items in a multidimensional setting, where each agent's valuation is represented by an $\ell$-dimensional vector. We explored two notions of relaxed envy-freeness tailored to this setting, which we call weak sEF$c$ and strong sEF$c$. We analyzed upper and lower bounds on $c$ that guarantee the existence of a weak sEF$c$ or a strong sEF$c$ allocation. In addition, we developed algorithms to determine whether a weak sEF$c$ or a strong sEF$c$ allocation exists, and we established NP-hardness results for checking the existence of weak sEF1 and strong sEF1 allocations.
Moreover, we demonstrated that determining whether a given allocation is strong sEF$c$ is NP-complete when $c$ is a parameter.

An important direction for future work is to narrow the gap between the upper and lower bounds on $c$ required to guarantee the existence of a weak sEF$c$ or strong sEF$c$ allocation. In particular, it remains an open question whether a weak sEF1 allocation always exists when the number of dimensions is two and the number of agents is at least $3$. Another promising direction is to explore the existence and computational complexity of other fairness and efficiency notions in multidimensional settings, such as proportionality.
For example, we can define multidimensional version of (relaxed) proportionalities as follows:
\begin{itemize}
\item An allocation $\bA$ is said to be \emph{weak simultaneously proportional up to $c$ goods (weak sPROP$c$)} if for any agent $i\in N$ and any dimension $k\in L$, there exists a set $X\subseteq M\setminus A_i$ such that $|X|\le c$ and $v_i(A_i\cup X)_k\ge v_i(M)_k/n$.
\item An allocation $\bA$ is said to be \emph{strong simultaneously proportional up to $c$ goods (strong sPROP$c$)} if for any agents $i\in N$, there exists a set $X\subseteq M\setminus A_i$ such that $|X|\le c$ and $v_i(A_i\cup X)_k\ge v_i(M)_k/n$ for any dimension $k\in L$.
\end{itemize}
Note that weak sPROP$c$ was introduced by Bu et al.~\cite{Bu2024fair}, but strong sPROP$c$ has not been studied.
For these notions of proportionality, the polytope-based approach would yield a meaningful guarantee; however, whether stronger guarantees can be achieved remains an interesting avenue for future work.

Finally, we believe that a similar analysis should be possible for the chore allocation case, but we leave this for future research.

\newpage
\bibliography{references}

\appendix

\section{Omitted proof}

\THMidExC*
\begin{proof}
Let $(N,M,L,(v_i)_{i\in N})$ be a given instance of the MDFA problem. 
We assume that $N=\{1,2\}$, $M=\{g_1,\dots,g_m\}$, $L=[\ell]$, and $v_1=v_2$.
Consider the following polytope in $\mathbb{R}^m$:
\begin{align}
P\coloneqq\left\{
x\in\mathbb{R}^m
\mid
\begin{array}{ll}
\sum_{j=1}^m v_{1jk}(2x_j-1)=0 & (\forall k\in [\ell]),\\[2pt]
0\le x_j\le 1                  & (\forall j\in [m])
\end{array}\!\!%
\right\}.
\end{align}
A feasible point $x\in P$ is a simultaneously envy-free fractional allocation, where $x_j$ is interpreted as the fraction of item $g_j$ allocated to agent 1.
Note that, since $x=(1/2,1/2,\dots,1/2)$ is in the polytope, so $P$ is not empty.

Since $P$ is $m$-dimensional, any extreme point is determined by $m$ linearly independent tight constraints.
In particular, at least $m-\ell$ of these tight constraints must be of the form $x_j=0$ or $x_j=1$.
This implies that every extreme point has at least $m-\ell$ coordinates that are integral (i.e., equal to $0$ or $1$). Equivalently, there can be at most $\ell$ fractional coordinates.

We now construct an allocation from an extreme point $x^*$ of $P$.
Note that an extreme point must exist since $P$ is nonempty and bounded.
Define
\begin{align}
&I_1\coloneqq \{g_j\in M\mid x^*_j=1\},
&&I_2\coloneqq \{g_j\in M\mid x^*_j=0\},\\
&F_1\coloneqq \{g_j\in M\mid 1/2\le x^*_j<1\},
&&F_2\coloneqq \{g_j\in M\mid 0< x^*_j<1/2\}.
\end{align}
Then, since all fractional coordinates of $x^*$ lie in $F_1\cup F_2$, we have $|F_1|+|F_2|\le \ell$.

Define the allocation $\bA=(A_1,A_2)$ by $A_1=I_1\cup F_1$ and $A_2=I_2\cup F_2$.
Since $x^*\in P$, for every dimension $k\in[\ell]$ we have
\begin{align}
v_1(A_1)_k
= \sum_{g_j\in A_1}v_{1jk}
&\ge \sum_{g_j\in A_1}v_{1jk}(2x^*_j-1)\\
&= \sum_{g_j\in A_2}v_{1jk}(1-2x^*_j)
\ge \sum_{g_j\in I_2} v_{1jk}
= v_1(A_2\setminus F_2)_k.
\end{align}
Similarly, by $x^*\in P$,
for every dimension $k\in[\ell]$ we have
\begin{align}
v_2(A_2)_k
= \sum_{g_j\in A_2}v_{2jk}
&\ge \sum_{g_j\in A_2}v_{2jk}(1-2x^*_j)\\
&= \sum_{g_j\in A_1}v_{2jk}(2x^*_j-1)
\ge \sum_{g_j\in I_1} v_{2jk}
= v_2(A_1\setminus F_1)_k.
\end{align}
Thus, any envy can be eliminated by removing all the fractional items from the bundle of the other agent. Hence, the allocation is strong sEF$\ell$.
Since $c\ge \ell$, the allocation is strong sEF$c$ as well.

Finally, since the polytope $P$ has polynomial number of constraints and variables, a basic optimal solution can be computed in polynomial time via standard linear programming techniques.
Hence, the resulting allocation $\bA$ can be found in polynomial time.
\end{proof}

\THMsch*
\begin{proof}
We first note that the problem is in NP. Indeed, for any distinct agents $i,i'$ and any candidate removal set $X\subseteq A_{i'}$ with $|X|\le c$, we can verify in polynomial time that for every dimension $k\in L$, it holds that $v_i(A_i)_k\ge v_i(A_{i'}\setminus X)_k$.

We now prove strong NP-hardness by reduction from the \emph{3-Dimensional Matching (3DM)} problem, which is known to be strongly NP-complete~\cite{garey1979computers}.
An instance of 3DM consists of three disjoint sets $X=\{x_1,\dots,x_n\}$, $Y=\{y_1,\dots,y_n\}$, $Z=\{z_1,\dots,z_n\}$, and a set $T\subseteq X\times Y\times Z$ of triplets.
The goal is to decide whether there exists a perfect matching $T'\subseteq T$ of size $n$, i.e., every element of $X \cup Y \cup Z$ appears in exactly one triple in $T'$.
Without loss of generality, we assume that $|T|>n$ since otherwise it is easy to verify existence of a perfect matching.

Given an instance of the 3DM problem, we construct an instance of the MDFA problem $(N,M,L,(v_i)_{i\in N})$ with an allocation $\bA$ and a positive integer $c$ as follows.
Let $N=\{1,2\}$, $M=T\cup\{g^*\}$ (i.e., one item from each triplet in $T$, plus a special item $g^*$), and $L=X\cup Y\cup Z$.
For every agent $i\in N$, item $t=(x,y,z)\in T~(=M\setminus\{g^*\})$, and dimension $w\in L$, define
\begin{align}
v_{i}(t)_w=\begin{cases}
1 & \text{if }w\in\{x,y,z\},\\
0 & \text{otherwise}.
\end{cases}
\end{align}
For the special item $g^*$, set $v_{i}(g^*)_w=1$ for all $i\in N$ and $w\in L$.
Since the valuation functions are defined in the same way for both agents, the instance is identical and binary.
In addition, we define an allocation $\bA=(\{g^*\},\,M\setminus\{g^*\})$, and we set the parameter $c=|T|-n~(\ge 1)$.

We now show that $\bA$ is strong sEF$c$ if and only if the 3DM instance is a yes-instance.

\medskip
\noindent\textbf{($\boldsymbol{\Rightarrow}$)}
Suppose that $\bA$ is strong sEF$c$. 
By definition, we must have
$$1=v_1(A_1)_w\ge v_1(A_2\setminus X)_w \qquad (\forall w\in L)$$ 
for some $Q \subseteq A_2$ with $|Q| \leq c=|T|-n$.
Define $T'=T\setminus Q$.
Since $|Q|\le c=|T|-n$, it follows that $|T'|=|T|-|Q|\ge n$.
Now, we prove that $T'=T\setminus Q$ is a perfect matching for the 3DM.

Suppose, for contradiction, that $T'$ is not a perfect matching.
Then, since $|T'|\ge n$, there exists $w^*\in X\cup Y\cup Z$ that appears at least twice in $T'$.
However, this is impossible because 
$$1=v_1(A_1)_{w^*}\ge v_1(A_2\setminus Q)_{w^*}=\big|\big\{t=(x,y,z)\in A_2\setminus Q\mid w\in\{x,y,z\}\big\}\big|\ge 2.$$ 
Thus, every $w\in X\cup Y\cup Z$ appears at most once in $T'$.
Since $|T'|\ge n$, every element $w\in X\cup Y\cup Z$ appears exactly once in $T'$. In other words, $T'$ is a perfect matching.

\medskip
\noindent\textbf{($\boldsymbol{\Leftarrow}$)}
Conversely, suppose that the given 3DM instance is a yes-instance; that is, there is a perfect matching $T'\subseteq T$.
Define the removal set $Q=T\setminus T'$. Then, we have
$$v_1(A_1)_w=1=\big|\big\{t=(x,y,z)\in T'\mid w\in\{x,y,z\}\big\}\big|=v_1(T')_w=v_1(A_2\setminus Q)_w$$
for any $w\in L$.
Thus, the envy from agent $1$ to agent $2$ can be eliminated by removing $|Q|=|T|-n=c$ items.
In addition, agent 2 does not envy agent 1 because 
$$v_2(A_2)_w\ge v_2(T')_w=\big|\big\{t=(x,y,z)\in T'\mid w\in\{x,y,z\}\big\}\big|=1=v_2(A_1)$$
for any $w\in L$.
Therefore, $\bA$ is a strong sEF$c$ allocation.

\medskip
Since the reduction can be performed in polynomial time and $\bA$ is strong sEF$c$ if and only if the 3DM instance is a yes-instance, checking whether a given allocation is strong sEF$c$ is strongly NP-hard even for the identical binary case.
\end{proof}

\end{document}